\documentclass[envcountsame]{llncs}     

\usepackage{amsmath}
\usepackage{amssymb,amsfonts,amsxtra}
\usepackage{stmaryrd}
\usepackage{pslatex}
\usepackage{graphicx}
\usepackage{epsfig}
\usepackage{mathrsfs}
\usepackage{latexsym,pifont,color}
\usepackage[english]{babel}
\newdimen\proofrulebreadth \proofrulebreadth=.05em
\newdimen\proofdotseparation \proofdotseparation=1.25ex
\newdimen\proofrulebaseline \proofrulebaseline=2ex
\newcount\proofdotnumber \proofdotnumber=3
\let\then\relax
\def\hfi{\hskip0pt plus.0001fil}
\mathchardef\squigto="3A3B
\newif\ifinsideprooftree\insideprooftreefalse
\newif\ifonleftofproofrule\onleftofproofrulefalse
\newif\ifproofdots\proofdotsfalse
\newif\ifdoubleproof\doubleprooffalse
\let\wereinproofbit\relax
\newdimen\shortenproofleft
\newdimen\shortenproofright
\newdimen\proofbelowshift
\newbox\proofabove
\newbox\proofbelow
\newbox\proofrulename
\def\shiftproofbelow{\let\next\relax\afterassignment\setshiftproofbelow\dimen0 }
\def\shiftproofbelowneg{\def\next{\multiply\dimen0 by-1 }%
\afterassignment\setshiftproofbelow\dimen0 }
\def\setshiftproofbelow{\next\proofbelowshift=\dimen0 }
\def\setproofrulebreadth{\proofrulebreadth}

\def\prooftree{
\ifnum  \lastpenalty=1
\then   \unpenalty
\else   \onleftofproofrulefalse
\fi
\ifonleftofproofrule
\else   \ifinsideprooftree
        \then   \hskip.5em plus1fil
        \fi
\fi
\bgroup
\setbox\proofbelow=\hbox{}\setbox\proofrulename=\hbox{}%
\let\justifies\proofover\let\leadsto\proofoverdots\let\Justifies\proofoverdbl
\let\using\proofusing\let\[\prooftree
\ifinsideprooftree\let\]\endprooftree\fi
\proofdotsfalse\doubleprooffalse
\let\thickness\setproofrulebreadth
\let\shiftright\shiftproofbelow \let\shift\shiftproofbelow
\let\shiftleft\shiftproofbelowneg
\let\ifwasinsideprooftree\ifinsideprooftree
\insideprooftreetrue
\setbox\proofabove=\hbox\bgroup$\displaystyle 
\let\wereinproofbit\prooftree
\shortenproofleft=0pt \shortenproofright=0pt \proofbelowshift=0pt
\onleftofproofruletrue\penalty1
}

\def\eproofbit{
\ifx    \wereinproofbit\prooftree
\then   \ifcase \lastpenalty
        \then   \shortenproofright=0pt  
        \or     \unpenalty\hfil         
        \or     \unpenalty\unskip       
        \else   \shortenproofright=0pt  
        \fi
\fi
\global\dimen0=\shortenproofleft
\global\dimen1=\shortenproofright
\global\dimen2=\proofrulebreadth
\global\dimen3=\proofbelowshift
\global\dimen4=\proofdotseparation
\global\count255=\proofdotnumber
$\egroup  
\shortenproofleft=\dimen0
\shortenproofright=\dimen1
\proofrulebreadth=\dimen2
\proofbelowshift=\dimen3
\proofdotseparation=\dimen4
\proofdotnumber=\count255
}

\def\proofover{
\eproofbit 
\setbox\proofbelow=\hbox\bgroup 
\let\wereinproofbit\proofover
$\displaystyle
}%
\def\proofoverdbl{
\eproofbit 
\doubleprooftrue
\setbox\proofbelow=\hbox\bgroup 
\let\wereinproofbit\proofoverdbl
$\displaystyle
}%
\def\proofoverdots{
\eproofbit 
\proofdotstrue
\setbox\proofbelow=\hbox\bgroup 
\let\wereinproofbit\proofoverdots
$\displaystyle
}%
\def\proofusing{
\eproofbit 
\setbox\proofrulename=\hbox\bgroup 
\let\wereinproofbit\proofusing
\kern0.3em$
}

\def\endprooftree{
\eproofbit 
  \dimen5 =0pt
\dimen0=\wd\proofabove \advance\dimen0-\shortenproofleft
\advance\dimen0-\shortenproofright
\dimen1=.5\dimen0 \advance\dimen1-.5\wd\proofbelow
\dimen4=\dimen1
\advance\dimen1\proofbelowshift \advance\dimen4-\proofbelowshift
\ifdim  \dimen1<0pt
\then   \advance\shortenproofleft\dimen1
        \advance\dimen0-\dimen1
        \dimen1=0pt
        \ifdim  \shortenproofleft<0pt
        \then   \setbox\proofabove=\hbox{%
                        \kern-\shortenproofleft\unhbox\proofabove}%
                \shortenproofleft=0pt
        \fi
\fi
\ifdim  \dimen4<0pt
\then   \advance\shortenproofright\dimen4
        \advance\dimen0-\dimen4
        \dimen4=0pt
\fi
\ifdim  \shortenproofright<\wd\proofrulename
\then   \shortenproofright=\wd\proofrulename
\fi
\dimen2=\shortenproofleft \advance\dimen2 by\dimen1
\dimen3=\shortenproofright\advance\dimen3 by\dimen4
\ifproofdots
\then
        \dimen6=\shortenproofleft \advance\dimen6 .5\dimen0
        \setbox1=\vbox to\proofdotseparation{\vss\hbox{$\cdot$}\vss}%
        \setbox0=\hbox{%
                \advance\dimen6-.5\wd1
                \kern\dimen6
                $\vcenter to\proofdotnumber\proofdotseparation
                        {\leaders\box1\vfill}$%
                \unhbox\proofrulename}%
\else   \dimen6=\fontdimen22\the\textfont2 
        \dimen7=\dimen6
        \advance\dimen6by.5\proofrulebreadth
        \advance\dimen7by-.5\proofrulebreadth
        \setbox0=\hbox{%
                \kern\shortenproofleft
                \ifdoubleproof
                \then   \hbox to\dimen0{%
                        $\mathsurround0pt\mathord=\mkern-6mu%
                        \cleaders\hbox{$\mkern-2mu=\mkern-2mu$}\hfill
                        \mkern-6mu\mathord=$}%
                \else   \vrule height\dimen6 depth-\dimen7 width\dimen0
                \fi
                \unhbox\proofrulename}%
        \ht0=\dimen6 \dp0=-\dimen7
\fi
\let\doll\relax
\ifwasinsideprooftree
\then   \let\VBOX\vbox
\else   \ifmmode\else$\let\doll=$\fi
        \let\VBOX\vcenter
\fi
\VBOX   {\baselineskip\proofrulebaseline \lineskip.2ex
        \expandafter\lineskiplimit\ifproofdots0ex\else-0.6ex\fi
        \hbox   spread\dimen5   {\hfi\unhbox\proofabove\hfi}%
        \hbox{\box0}%
        \hbox   {\kern\dimen2 \box\proofbelow}}\doll%
\global\dimen2=\dimen2
\global\dimen3=\dimen3
\egroup 
\ifonleftofproofrule
\then   \shortenproofleft=\dimen2
\fi
\shortenproofright=\dimen3
\onleftofproofrulefalse
\ifinsideprooftree
\then   \hskip.5em plus 1fil \penalty2
\fi
}

\usepackage{proof}

\newcommand{\act}[1]{\mathsf{actor}\{#1\}}

\newcommand{\rea}[2]{\mathsf{react}\{#1\Rightarrow{#2}\}}

\newcommand{\vx}{{\ensuremath{\tilde{x}}}}
\newcommand{\vy}{{\ensuremath{\tilde{y}}}}

\newcommand{\vu}{{\ensuremath{\tilde{u}}}}
\newcommand{\va}{{\ensuremath{\tilde{a}}}}
\newcommand{\vb}{{\ensuremath{\tilde{b}}}}
\newcommand{\vc}{{\ensuremath{\tilde{c}}}}
\newcommand{\vT}{{\ensuremath{\tilde{T}}}}

\newcommand{\val}{\mathsf{val}\,}
\newcommand{\self}{\mathsf{self}}

\newcommand{\s}{\,!\,}

\newcommand{\eqdef}{\stackrel{\triangle}{=}}

\newcommand{\red}{\longrightarrow}

\newcommand{\para}{~|~}

\newcommand{\Nu}{\boldsymbol{\nu}}
\newcommand{\New}[1]{(\Nu #1)}

\newcommand{\fn}[1]{\mathrm{fn}(#1)}

\newcommand{\bn}[1]{\mathrm{bn}(#1)}

\renewcommand{\b}[1]{\mathbf{#1}}
\newcommand{\dead}{{\b{0}}}

\newcommand{\fineST}{\mathsf{end}}

\newcommand{\gv}{\Gamma\vdash}
\newcommand{\escape}{\triangleright\ }

\newcommand{\marked}[1]{#1^\bullet}
\newcommand{\fullym}[1]{\mathsf{fullmrk}(#1)}

\newcommand{\escS}{\mbox{{\large $\&$}}}

\newcommand{\mergesuff}{\subsetplus}
\newcommand{\after}[1]{\lfloor_{#1}}

\newcommand{\actors}[1]{\mathsf{actors}(#1)}

\newcommand{\labeltree}[3]
{\begin{tabular}{l}
 \textsc{#1} \\[1mm]
 {\prooftree
 \text{$#2$}
 \justifies
 \raisebox{-1mm}{\text{$#3$}}
 \endprooftree}
 \end{tabular}
}

\newcommand{\labeltreeside}[4]
{\begin{tabular}{l}
 \textsc{#1} \\[1mm]
 {\infer[#4]{#3}{#2}}
 \end{tabular}
}

\newcommand{\escapes}{\triangleright}

\newcommand{\sseq}{<\!\!<}

\newcommand{\balanced}{\mathsf{balanced}}

\begin{document}
\title{Behavioural Types for Actor Systems}
\author{Silvia Crafa}

\institute{
Dipartimento di Matematica - Universit\`a di Padova
}
\maketitle

\begin{abstract}
Recent mainstream programming languages such as Erlang or Scala
have renewed the interest on the Actor model of concurrency. 
However, the literature on the static analysis of actor systems is
still lacking of mature formal methods.
In this paper we present a minimal actor calculus that takes as
primitive the basic constructs of Scala's Actors API. More precisely,
actors can send asynchronous messages, process received messages
according to a pattern matching mechanism, and dynamically create new
actors, whose scope can be extruded by passing actor names as message
parameters. 
Drawing inspiration from the linear types and session type theories
developed for process calculi, we put forward a behavioural type
system that addresses the key issues of an actor calculus.
We then study a safety property dealing with the
determinism of finite actor communication. 
More precisely, we show that well typed and \emph{balanced} actor 
systems are ($i$) deadlock-free and ($ii$) any message will eventually
be handled by the target actor, and dually no actor will indefinitely
wait for an expected message. 
\end{abstract}

\section{Introduction}
Recent mainstream programming languages such as Erlang or Scala
have renewed the interest on the Actor model
(\cite{Hewitt73,ActorsFirst}) of concurrent and 
distributed systems. In the Actor model a program is an ensemble of
autonomous computing entities communicating through asynchronous
message passing. Compared to shared-state concurrent processes, the
Actor model more easily avoids concurrency hazards such as data races
and deadlocks, possibly at the cost of augmenting the communication
overhead. On the other hand, compared to the channel-based 
communication of process calculi such as the $\pi$-calculus or the
Join-calculus, the actor abstraction better fits 
the object oriented paradigm found in mainstream programming
languages.

Actors can send asynchronous messages, process received messages
according to a pattern matching mechanism and dynamically create new
actors, whose scope can be extruded by passing actor names as message
parameters. The Actor model and the asynchronous process calculi then 
share similarities such as (bound) name passing, but they have
also many differences: actors have an identity (a name), they are
single threaded and they communicate by sending messages to the
mailbox of other actors rather than using channels. 
Despite their similarities, while a rich literature on type-based
formal methods has been developed for the 
static analysis of process calculi, few works deal with the Actor
model (see Section~\ref{sec:conclude} for a discussion of the related
work). In this paper we study a minimal actor calculus, AC, 
with the aim of bringing in the context of Actors the successful
techniques developed for process calculi. 

More precisely, drawing inspiration from the linear types and session
type theories developed for process calculi
(\cite{THK94,HVK98,KPT99,GayV10}), we put forward a type 
system that addresses the key issues of an actor calculus.
Programming an actor system entails the design of a communication
protocol that involves a (dynamic) set of actors;
we then study a behavioural type system for AC where actor types
encode the intended communication protocol, and the type checking
phase statically guarantees that runtime computation correctly
implements that protocol. Moreover, we study a safety property dealing
with the determinism of actor communication, so that 
in well typed and \emph{balanced} actor systems any message will
eventually be processed by the target actor, and viceversa, no actor
will indefinitely wait for an expected message.  
Dealing only with finite computation, we devise a simple technique to
let types also prevent deadlocks.

Even if AC only considers actors with finite computation, which is
clearly a strong limitation, proving that a finite system complies
with the intended communication protocol is not trivial, since
nondeterminism, fresh actor name passing and the asynchronous
semantics of the underlying model complicate the picture. 
As we said, our behavioural type system is reminiscent of the type
discipline of linear and session types. However, even if the actor
calculus shares with session types the idea of conceiving the
computation as the implementation of a specified communication
protocol, there are a number of key differences between the two
models (see Section~\ref{sec:conclude}). 
In~\cite{FeatErlang} session types are added to a Erlang-style core
actor calculus. However, in that paper session types appears an
orthogonal feature of the actor language, while we aim at 
showing in this paper that the reasoning underlying session types is
in some sense inherent to the Actor model. 
As a general comment that can guide the reader through the
technical part of the paper, one can think that session
types describe the flow of communications withing a single conversation
session. Instead, an actor's behavioural type takes the point of view
of an entity that might concurrently participate to different
(interleaved) conversations with different parties.

\section{The Actor Calculus}
\label{sec:AC}
We assume a countable set of actor names and a countable set of
variables, ranged over by $a,b,c$ and $x,y$
respectively. Identifiers, denoted with $u$, rage over names and
variables. We reserve the letter $m$ to range 
over a distinct set of message labels. 
The syntax of the Actor calculus AC comes from the basic
constructs of Scala's Actors API~\cite{ScalaBook,ScalaActors}:
$$
\begin{array}{rll}
\mathit{Expressions}\ \ e & ::= & \dead ~|~
            u \s m(\vu);e ~|~ \rea{m_i(\vx_i)}{e_i}_{i\in I}  
        ~|~\val a=\act{e};e
\end{array}
$$
The expression $u_1\s m(\vu_2);e$ sends to the actor $u_1$ the message
$m$ with the tuple of actual parameters $\vu_2$, and then continues as
$e$. According to the Actor model, sending a message is an asynchronous
action that just adds the message $m(\vu_2)$ into the mailbox
associated to the actor $u_1$. Message handling is carried over by the
$\mathsf{react}$ expression, that suspends the execution of the actor
until it receives a message $m_j\in\{m_i\}_{i\in I}$. When a matching
message is found in the actor mailbox, the execution is resumed and
the corresponding continuation is activated. 

New actors are dynamically created with expression 
$\val a=\act{e_1};e_2$, that corresponds to the
inline Scala primitive for actor creation. It defines and starts a new
actor with name $a$ and body $e_1$ and then continues as $e_2$.
The actor definition introduces a new, bound, name $a$ whose scope is
both the new actor's code $e_1$ and the continuation $e_2$.
In order to have a uniform semantics, we assume that a \emph{program}
is a top level sequence of actor definitions, 
while input and output expressions can only occur inside an actor
body. 
This is not restrictive since it would be
sufficient to assume an implicit main actor containing
the top level sequence of expressions. 
Anyway, observe that besides the top level definitions, new actors can
still be dynamically spawned by 
other actors anytime during the computation. 


%

The execution of a program  
spawns a bunch of concurrent actors that interact by message
passing. Therefore programs are represented runtime by 
configurations:
$$
\begin{array}{rll}
\mathit{Configurations}\ \ F & ::= & \dead ~|~
                 [a\mapsto M]\,a\{e\} ~|~ 
                   F|F   ~|~ \New{a}F  ~|~ 
                   e  
\end{array}
$$
$\New{a}F$ is a configuration where $a$ is a private actor name. 
While each actor is single threaded, a configuration is a
parallel composition of a number of active actors and an expression
$e$ containing the residual sequence of top level actor definitions.
An \emph{active actor} $a$ is represented runtime by $[a\mapsto
M]\, a\{e\}$, where $e$ is the residual body of the actor and
$[a\mapsto M]$ is its associated mailbox.
Mailboxes are lists of received messages of the form
$[a\mapsto m_1(\vb_1)\cdot\ldots\cdot m_k(\vb_k) ]$.
Message parameters are values (i.e. actor names)
since, according to Scala semantics, message parameters are
called by value as they are implemented as parameters of an 
Actor object's method invocation (\cite{ScalaBook}).

\begin{definition}[Free Names and Well Formed Configurations]
In input expressions 
formal parameters are bound variables, 
and actor definitions act as name binders.
We work with well formed configurations where
any bound name and variable is assumed to be distinct (Barendregt's
convention), and where 
in any input branching  $\rea{m_i(x_i)}{e_i}_{i\in I}$ 
the labels $m_i$ are pairwise distinct.
\end{definition}

\begin{figure}[t]
{\small
\begin{center}
\begin{tabular}{c}
\labeltree{(Par)}
{ F_1 \red F_1'}
{F_1\para F_2 \red F_1'\para F_2}
\quad\quad
\labeltree{(Res)}
{F \red F'}
{\New{a}F \red \New{a} F'}
\quad\quad
\labeltree{(Ended)}
{}
{[a\mapsto\varnothing]\, a\{\dead\} \red \dead}
\\ \\
\labeltree{(Struct)}
{F\equiv F'\ \red F'' \ \equiv F'''}
{F\red F'''}
\quad\quad
$\begin{array}{c}
\New{a}\New{b}F  \equiv  \New{b}\New{a}F\\
\New{a}(F\para F')  \equiv   F\para \New{a}F' \ \ \
a\notin\fn{F}\\
F\para \dead \equiv F \ \ \ F\para F'\equiv F'\para F\\
(F_1\para F_2)\para F_3\equiv
F_1\para(F_2\para F_3)\ \ \ \New{a}\dead\equiv\dead
\end{array}$
\\ \\
\labeltree{(top spawn)}
{}
{\val a=\act{e};e' 
 \red   \New{a}([a\mapsto\varnothing]\, a\{e\}\}\para e'))}
\\ \\
\labeltreeside{(spawn)}
{}
{[b\mapsto M]\, b\{\val a=\act{e};e'\} 
  \red \New{a}( [b\mapsto M]\, b\{e'\} \para
  [a\mapsto\varnothing]\, a\{e\}))}{a\notin \fn{M}}
\\ \\
\labeltree{(Send)}
{}
{[a\mapsto M]\, a\{e\}\,\para\, [b\mapsto M']\, b\{a\s m(\vc);e'\} \red 
 [a\mapsto M\cdot m(\vc)]\, a\{e\}\,\para\, [b\mapsto M']\,b\{e'\}}
\\ \\
\labeltreeside{(Receive)}
{}
{[a\mapsto M\!\cdot\! m_j(\vc)\!\cdot\! M']\ 
 a\{ \rea{m_i(\vx_i)}{e_i}_{i\in I}\} \red  
  [a\mapsto M\!\cdot\! M']\ a\{e_j\{^\vc/_{\vx_j}\} \}}{j\in I}
\vspace*{-2mm}
\end{tabular}
\end{center}
\caption{Operational semantics}
\label{fig:opsem1}
}
\end{figure}

The operational semantics is given in
Figure~\ref{fig:opsem1}. Most of the rules come directly from the
$\pi$-calculus. 
The rule {\sc (Ended)} states that a terminated actor $a$ with no
pending message in its mailbox can be garbage collected.
%
%
The rules {\sc (top spawn)} and {\sc (Spawn)} are used to spawn
a new actor respectively from the top level main thread and from
another actor. In both cases the new actor is activated by extending
the configuration with a new empty mailbox and an additional 
thread running the body of the new actor. 
The rules {\sc (Send)} and {\sc (Receive)} implement the Actor
communication model: an output expression adds a message to the
mailbox of the target actor, while an input expression scans the
mailbox for a matching message. Notice that the mailbox is not handled
as an ordered queue of messages, hence for instance, the 
configuration (where we omit message parameters)
$$
\begin{array}{ll}
[b\mapsto\varnothing]\ a\{\, b\s m_1;b\s m_2;\dead\}~|~\para
[b\mapsto\varnothing]\ b\{\, \mathsf{react} & \{m_1\Rightarrow
\rea{m_2}{e},
\\
& \ m_2\Rightarrow \rea{m_1}{e'}\}\}
\end{array}
$$
nondeterministically reduces either to $[b\mapsto\varnothing]b\{e\}$
or to $[b\mapsto\varnothing]b\{e'\}$.
In other words, besides being asynchronous, in the Actor model the
ordering of outputs is not guaranteed to be mirrored by the 
ordering of input handlers, which is instead the case of, e.g.,
asynchronous session types with buffered channels
~\cite{GayV10,CDYAsynch07}.

\begin{example}\label{ex:PingPongPang}
The following program defines two actors that meet in a three-way
handshake. The actor $b$ starts by sending a $ping$ message to $a$,
then waits for a $pong$ message that carries the name of the actor to
which it sends the final $pang$ message. The actor $a$ performs the
dual sequence of actions. 
$$
\begin{array}{l}
Pr = 
\val a=\act{\rea{ping(x)}{x\s pong(a);\rea{pang()}{\dead}}}\ ;
\\[2mm]
\quad\quad
\val b=\act{a\s ping(b);\rea{pong(y)}{y\s pang();\dead}}\ ;\ \dead
\end{array}
$$
%

\noindent
Now consider the case where the actor $Alice$ starts two sessions
of this protocol to interact both with $Bob$ and $Carl$ (Figure~\ref{fig:examples}).
In order to prevent interferences between the two sessions,  
a couple of private sub-actors are established for each protocol
session. This is similar to private sessions in the
$\pi$-calculus. 
$$
\begin{array}{l}
  Alice\ \{\ \val ab=\act{\rea{dest(y)}{P(y)}}\,;\,
               Bob\s new(ab)\,;\,
\\[2mm]
     \quad \quad\quad 
\val ac=\act{\rea{dest(y)}{P(y)}}\,;\,
               Carl\s new(ac)\,;\,\dead\  \} \ \para
\\[2mm]
  Bob\ \{\ \rea{new(z)}{\val ba=\act{Q(z)};
   z\s dest(ba);\dead}\  \} \ \para
\\[2mm]

 Carl\ \{\ \rea{new(z)}{\val ca=\act{Q(z)};
   z\s dest(ca);\dead}\  \}
\end{array}
$$
where 
$P(y) =  y\s ping.\rea{pong}{y\s pang;\dead}$
and\\ $Q(z)=\rea{ping}{z\s pong.\rea{pang}{\dead}}$.
\end{example}

\begin{example}
We can rephrase in the actor calculus a simple
example of multiparty communication protocol that captures the
interactions in a purchase system (Figure~\ref{fig:examples}):
$$
\begin{array}{l}
Buyer\{\, Seller\s buy(Buyer,item);\, 
                 \rea{price(z)}{\rea{details(w)}{...}}\,\}~\para
\\[2mm]
Seller\{\, \mathsf{react}\{\, buy(x,y)\Rightarrow x\s price(f(y));
\\[1mm]\hspace{3.3cm} 
    \val Shipper=\act{\rea{ship(x,y)}{x\s details(f'(y));...}};
    \\[1mm]   \hspace{3.3cm}         Shipper\s ship(x,y);... \};
\end{array}
$$
A $Buyer$ actor sends to the $Seller$ actor its name together with the
item he wants to buy, and waits for the price and the shipping
details. Dually, the $Seller$ handles the $buy$ message by sending to
the $Buyer$ the price $f(item)$ of the selected item and spawns a new
$Shipper$ actor that directly interacts with the $Buyer$ to finalize
the shipping. Observe that the $Buyer$ actor needs not to be aware
that he is actually interacting not only with the $Seller$ but also
with a (restricted) $Shipper$. This is a further difference with the
case of multiparty session types, where each interacting party is
identified by its endpoint of the session channel.
\end{example}

\begin{figure}[t]
\begin{center}
\begin{tabular}{cc}
\includegraphics[width=5.3cm]{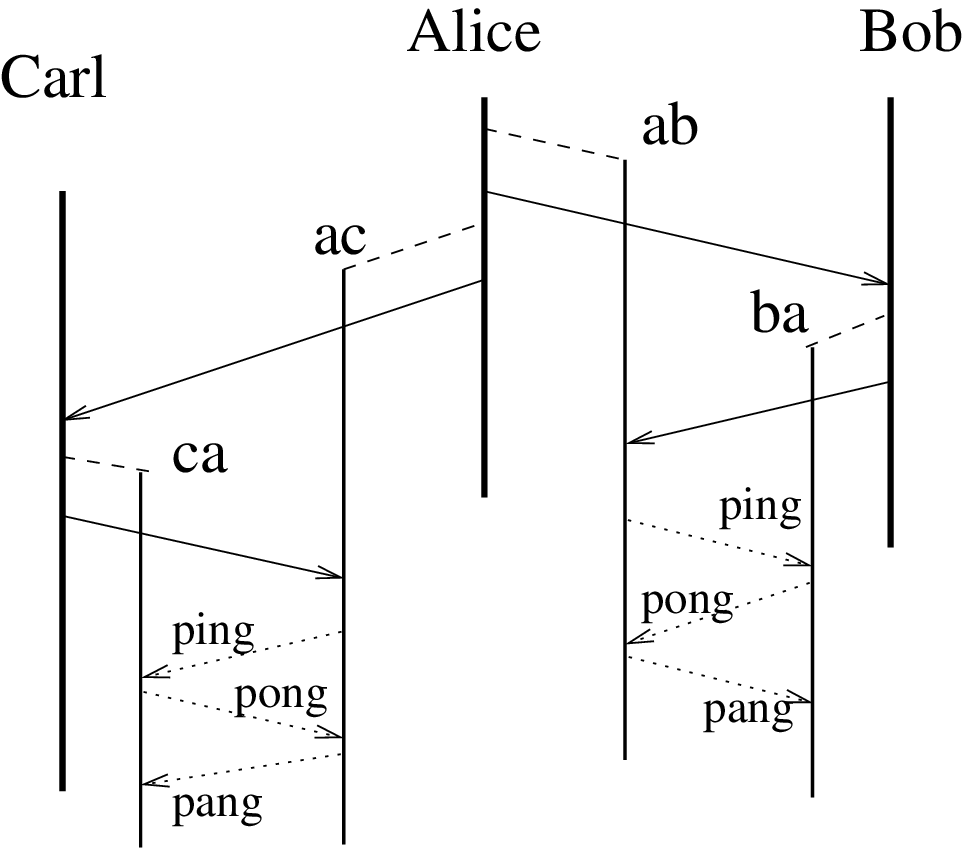}
\hspace{5mm}
&
\hspace{5mm}
\includegraphics[width=3.6cm]{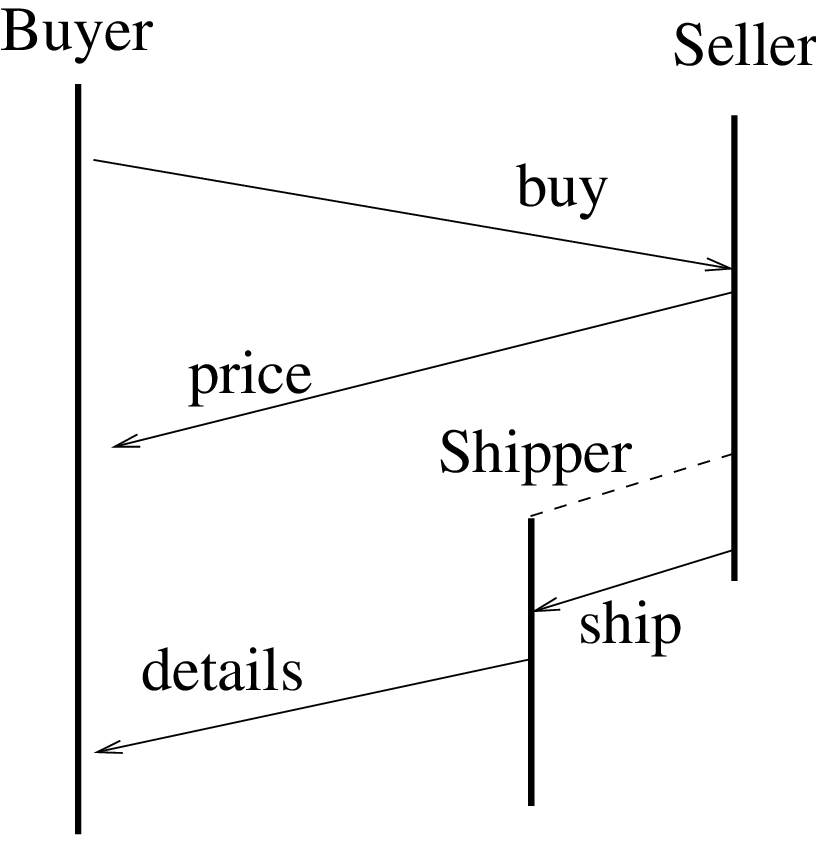}
\end{tabular}
\end{center}
\label{fig:examples}
\caption{Examples}
\end{figure}

\section{The Type System}

We assign behavioural types to
actor names so that a type describes the sequence of
inputs and outputs performed by the actor body. 
Moreover, inputs  are handled as linear resources, so that to
guarantee that for each expected input there is exactly one matching
output. 
We use the following syntax for the types associated to actor names,
where $\mathsf{NoMark}(S)$ means that $S$ does not contain any
marking: 
$$
\begin{array}{ll}
\mathit{Types}\ \ T ::=  [S]  \quad\quad & 
S  ::=  \fineST ~|~ !m(\vT).S 
  ~|~\&_{i\in I}\{?m_i(\vT_i).S_i\} 
\\[2mm]
   &  \ \ \ 
  ~|~  \marked{\&_{i\in I}} \{\marked{}?m(\vT).S,\ ?m_i(\vT_i).S_i\}
   \ \mbox{ with } \mathsf{NoMark}(S_i) \ \forall i\in I 
\end{array}
$$
Types $T$ are finite sequences of input and output actions.
Output action $\s m(\vT)$ is the type of an output expression that
sends the message $m$ with a tuple of parameters of type $\vT$. 
Dually, the input action $\&_{i\in I}\{?m_i(\vT_i).S_i\}$ offers the
choice of receiving one of the messages $m_i$ and continuing with
the sequence $S_i$. Differently from session types, we do not consider
output choices. Indeed our aim is not to provide an expressive
calculus for protocol specification, but to put forward a type based
technique to statically verify the protocol conformance of actors.
On the other hand, it would not be difficult to extend the type system
with output selection $\oplus_{i\in I}\{\s m_i(\vT_i).S_i\}$ along the
lines of input branches.

The type system then makes use of linear type assumptions to guarantee
that each input is eventually matched by exactly one output in the
system. Linear type assumptions are handled by means of markings. 
The marked action $\marked{\&_{i\in I}} \{\marked{}?m(\vT).S,\
?m_i(\vT_i).S_i\}$ pinpoints an input that is ``consumed'' by
one output expression.  
To illustrate, the actor $a\{b\s m(c)\}$ is well typed  
assuming $a{:}[\s m(T).S_a],\, b{:}[\marked{}? m(T).S_b]$ 
since $a$ consumes $b$'s input. On the other hand
the actor $b\{\rea{m(x)}{e}\}$ is well
typed assuming the non marked type $b:[? m(T).S_b]$, since $b$
offers an input without consuming it.
Moreover, to deal with branching inputs we have to ensure that
all the messages eventually received by an actor belong to the same
branch of computation. 
For instance, consider the actor $a:[\&\{?m_1.?m_2,\ ?m_3.?m_4\}]$
(where we omit message parameters), then the actor 
$b\{a\s m_1;a\s m_4\}$ is incorrect since it sends to $a$ two messages
belonging to alternative execution paths. Indeed, the typing of $b$
would require for $a$ the type assumption 
$a:[\marked{\&}\{\marked{}?m_1.?m_2,\ ?m_3.\marked{}?m_4\}]$,
which is prohibited by our syntax of types since it
contains two markings in two different branches. 

Another key point is the parallel composition of type assumptions,
that must be defined so that to ensure the linear usage of marked
inputs. To illustrate, given the type of the actor $a$ above, the
parallel composition $b\{a\s m_1\}\para c\{a\s m_3\}$ must be
prohibited since only one of the two messages will be handled by $a$
while the other one will stay pending in $a$'s mailbox. In other
terms, the two outputs compete for the same ``input resource''. 
Notice that the typing of $b$, resp. $c$, would require the assumption
$a:T_a^1=[\marked{\&}\{\marked{}?m_1.?m_2,\ ?m_3.?m_4\}]$,
resp. $a:T_a^2=[\marked{\&}\{?m_1.?m_2,\ \marked{}?m_3.?m_4\}]$. 
The fact that the same input choice is marked both in $T_a^1$ and
$T_a^2$ indicates that $b$ and $c$ consume the same input
choice, hence they cannot be composed in parallel.

More formally, we define a \emph{merge-mark} function that is used to
linearly compose type assumptions. More precisely, parallel threads
must assume the same type assumptions but with disjoint markings.

\begin{definition}[Merge-Mark]
Let $S,S'$ be two sequences that are equal
but for the markings. Then 
$[S]\uplus[S']{=}[S\uplus S']$ where the partial
function $\uplus$ is defined by
$$
\begin{array}{l}
\fineST\uplus\fineST  = \fineST
\quad\quad\quad\ \ 
\s m(\vT).S \uplus \s m(\vT).S'  =   \s m(\vT).(S\uplus S')
\\[1mm]
\&_{i\in I}\{?m_i(\vT_i).S_i\}\uplus \&_{i\in I}\{?m_i(\vT_i).S'_i\}  = 
\&_{i\in I}\{?m_i(\vT_i).(S_i\uplus S'_i)\}
\\[1mm]
\&_{i\in I}\{?m(\vT).S,?m_i(\vT_i).S_i\}\uplus \marked{\&}_{i\in
  I}\{\marked{}?m(\vT).S',?m_i(\vT_i).S_i\}  \ = \quad\quad\hfill\quad
\\[1mm] 
\hfill = \ \marked{\&}_{i\in I}\{\marked{}?m(\vT).(S\uplus S'),?m_i(\vT_i).S_i\}
\end{array}
$$
\end{definition}
The main clause is the last one, together with the symmetric
one that we omit. A marked input can only be merged with a
corresponding non marked input, and merging recursively applies
only to the marked branch. In this way we ensure that the same input
choice is consumed by exactly one output, and that further outputs 
only consume inputs belonging to the same branch of computation.

A type environment $\Gamma$ is a partial function assigning types
to names and variables. We use for $\Gamma$ the list notation.
Let be $\Gamma_1, \Gamma_2$ two type environments such that
$Dom(\Gamma_1)=Dom(\Gamma_2)$, then we denote by 
$\Gamma_1\uplus\Gamma_2$ the type environment obtained 
by merging the markings contained in the two environments, i.e.,
$\Gamma_1\uplus\Gamma_2=\{u:\Gamma_1(u)\uplus\Gamma_2(u) ~|~ u\in Dom(\Gamma_1)=Dom(\Gamma_2)\}$.  We use the notation
$\Gamma;a:T$ for environment update, that is 
$\Gamma\setminus\{a:\Gamma(a)\}\cup\{a:T\}$.

So far so good, however this is not enough since the scope extrusion
mechanism obtained by passing fresh actor names as message parameters
raises additional issues. Consider the actor
$a\{\rea{foo(x)}{x\s m(a)}\}$, $a$ consumes the 
$m$ input offered by some actor which will be dynamically
substituted for the bound variable $x$. In order to statically collect
the resources consumed by $a$, the typing of $a$ must assume 
for $x$ the marked type $[\marked{}?m(T).\fineST]$. 
A similar situation applies when new actors are spawned. For instance,
consider the previous actor $a$ in parallel with 
$c\{\val b=\act{\rea{m(y)}{e}};a\s foo(b) \}$. The parameter $x$ of
the $foo$ message is substituted with the fresh actor name $b$.
Hence in order to check that every input in the system is consumed,
besides the type of $x$, we must record the type of the fresh actor
$b:[?m(T).\fineST]$ and devise a way of matching the input
``consumed'' in $x$ with that ``offered'' by $b$.

We then rely on type judgements of the form $\gv F\escape\Delta$,
where $\Gamma$ collects the type assumptions about free names
and variables of $F$, while $\Delta$ collects type assumptions on
bound names and bound variables of $F$. Observe that working under
Barendregt's convention on bound names/variables, 
we avoid name conflicts.
We call $\Delta$ the \emph{escape environment},
and we let it preserve the branching structure of the computation
where alternative continuations can be activated when an input choice
is resolved.

\begin{definition}[Escape Envirnoment]
The escape environment $\Delta$ is a choice between alternative 
type environments defined by $\Delta::=\escS_{i\in I}\Gamma_i\para
\escS_{i\in I}\Delta_i$, where  $Dom(\Gamma_i)\cap
Dom(\Gamma_j)=\varnothing$ and $Dom(\Delta_i)\cap
Dom(\Delta_j)=\varnothing$ for $i,j\in I$.
\end{definition}
\noindent
We use the following notation for escape environment extension:
$(\escS_{i\in I}\Delta_i),\, u:T\eqdef\escS_{i\in I}\,
(\Delta_i,u{:}T)$, and 
$(\escS_{j\in J}\Delta_j)\, ,\, (\escS_{i\in
  I}\Delta_i)\eqdef\escS_{j\in J}\escS_{i\in I} \, (\Delta_j,\Delta_i)$.

\subsection{Typing Rules} 


The main judgement of the type system is 
$\Gamma\vdash F\escapes\Delta$. It means that 
the actors in $F$ execute the sequence of actions described 
by their type and the marked input actions in $\Gamma$ and $\Delta$
are \emph{exactly} those that are consumed by the actors in $F$. 
Moreover, $Dom(\Gamma)\cap Dom(\Delta)=\varnothing$ and
$\fn{F}\subseteq Dom(\Gamma)$ while $\bn{F}\subseteq Dom(\Delta)$.
We also use additional judgements: $\gv\diamond$ states that $\Gamma$
is well formed (according to standard rules given in
Figure~\ref{fig:TypeActor}), $\Gamma\vdash a:T$ 
states that the actor $a$ has type $T$ in $\Gamma$, and 
$\Gamma\vdash[a\mapsto M]$ states that the mailbox only
contains messages that are well typed according to the type that
$\Gamma$ assigns to $a$.
Finally, $\Gamma\vdash_a e\escapes\Delta$ states that $e$ is well
typed as the body of the actor $a$.

\paragraph{Type rules for actors.}

\begin{figure}[t]
{\small
\begin{center}
\labeltree{(Type Identif)}
{\Gamma, u:T\vdash\diamond}
{\Gamma, u:T\vdash u:T}
\quad\quad
\labeltree{(Ctx Empty)}
{}
{\varnothing\vdash\diamond}
\quad\quad
\labeltree{(Ctx Identif)}
{\Gamma\vdash\diamond\quad u\notin Dom(\Gamma)}
{\Gamma, u:T\vdash \diamond}
\end{center}
\begin{center}
\begin{tabular}{c}
\labeltree{(Type Spawn)}
{
\Gamma_1,b:[S_1]\vdash_b e_1 \escape\Delta_1
\quad
\Gamma_2,b:[S_2]\vdash_a e_2\escape\Delta_2
\quad
b\notin Dom(\Delta_1,\Delta_2)
}
{\Gamma_1\uplus\Gamma_2\vdash_a \val b=\act{e_1}; e_2 \escape
  \Delta_1,\Delta_2,\{b:[S_1\uplus S_2]\} }
\\ \\
\labeltree{(Type Send)}
{
\begin{array}{l}
 \gv a: [\s m(\vT).S_a] \quad \quad
\gv u: [S_u\,.\,\marked{\&}_{i\in I}\{\marked{}?m(\vT).S,\ ?m_i(\vT_i).S_i\}]
\\[2mm]
\vT\sseq \Gamma(\vu')\after{m}
\quad
 \Gamma\,;\, u: [S_u.\&_{i\in I}\{?m(\vT).S,\ ?m_i(\vT_i).S_i\}] \,;\, 
   a:[S_a]\vdash_a e\escape{\Delta}  
 \end{array}}
{\Gamma\vdash_a u\s m(\vu'); e \escape \Delta}
\\ \\
\labeltree{(Type Receive)}
{\gv a: [ \&_{i\in I} \{?m_i(\vT_i).S_i\}]\quad\quad 
\Gamma\, ;\, a:[S_i], \vx_i:\vT_i\vdash_a e_i
\escape\Delta_i 
\quad  i\in I
}
{\Gamma\vdash_a \rea{m_i(\vx_i)}{e_i}_{i\in I}\escape \escS_{i\in I} \,
  (\Delta_i,\{\vx_i:T_i\})}
\quad\quad 
\labeltree{(Type End)}
{\gv a: [\fineST] \quad \mathsf{NoMark}(\Gamma)}
{\Gamma\vdash_a \dead\escape\varnothing}
\vspace{-2mm}
\end{tabular}
\end{center}
}
\caption{Type Rules for Actors}
\label{fig:TypeActor}
\end{figure}

The rule {\sc (Type Spawn)} applies when the actor $a$ spawns a new
actor $b$. The type assumptions are split between those used by the
continuation of $a$'s body $e_2$, and those used by the body of the
new actor $e_1$. The same holds for escape assumptions, which 
collect the resources offered and consumed by bound names and
variables in $e_1$ and $e_2$. 
The name $b$ of the new actor must be fresh, and a type for $b$ must
be guessed. Since the scope of $b$ includes both
$e_1$ and $e_2$, both expressions are typed under a
suitable assumption for $b$. $S_1$ must correctly describe the
sequences of actions performed by $b$'s body $e_1$. 
Moreover, $S_1$ must contain marked input actions that correspond
to the input offered by $b$ and locally consumed by messages sent by
$b$'s body to $b$ itself. On the other hand, the marked inputs in
$S_2$ must correspond to the messages sent to $b$ in $e_2$. 
In the conclusion of the rule the escape environment globally collects
the type assumptions locally used for $b$, hence $S_1\uplus S_2$ must
be defined, that is $S_1$ and $S_2$ must be the same sequence of
actions with disjoint markings.


\noindent
Accordingly to {\sc (Type Send)}, when $a$ sends
the message $m(u')$ to the actor $u$:
\begin{enumerate}
\item the first action in the type of $a$ is the output of $m$;
\item the type of $u$ contains the input of $m$ as a marked input. 
  The matching input needs not to be the first action in the type 
  of $u$. This allows for instance that even if
  $u$ accepts a $\mathit{foo}$ message before  $m$, the actor
  $a$ is free to first output $u\s m$ and then $u\s \mathit{foo}$,
  according to the semantics of AC.
\item The continuation $e$ is typed in an updated environment 
where the marking of the matching input has disappeared from the
type of $u$ to record the fact that the resource has been already
consumed. Moreover the type of $a$ is updated to the continuation type
$[S_a]$ to record that the output action has been already performed.
Observe that this implies that the behavioural type assumed for an
actor changes (decreases) as long as the actor advances in its
computation.   
\end{enumerate}
As far as the typing of the message parameters are concerned,
let first introduce some notation: given two sequences $S$ and $S'$,
we write $S\sseq S'$ when $S$ is a suffix of $S'$ 
independently of the markings 
(see Appendix~\ref{app:def} for the formal definition). 
Given the output $u\s m(\vu')$, it might not be the case that
the actual parameters $\vu'$ have the types of the formal parameters
$\vT$. 
Since the type of an actor decreases as long as the actor computes, in
the asynchronous semantics the type of an actual parameter $u'$ at
sending time can be different from (longer than) the type $u'$ has
when $u$ processes the message. Then the type of a formal parameter
$T$ is in general a suffix of the type $\Gamma(u')$. Moreover, the
marked inputs 
contained in $\Gamma(u')$ are those that are consumed by the body of
$a$, while the markings contained in
$T$ correspond to the inputs offered by the formal parameter and
consumed by the actor that receives the message, that in general is
not $a$.
Summing up, the rule {\sc (Type Send)} requires (each type in the
tuple) $\vT$ to be a suffix of $\Gamma(\vu')$ (componentwise),
independently of the marked
actions. A stronger requirement is needed when a parameter of the
message coincides with the sender $a$, resp. the receiver $u$. 
In these cases the type of that formal parameter must be a suffix 
of the residual type of $a$ after the output, resp. the residual type
of $u$ after the input.
The rule uses the following predicate (componentwise extended to tuples
of types), where we call
$\Gamma(u')\after{m}$ the type of $u'$ ``after $m$'': 
$$
\begin{array}{ll}
T\sseq \Gamma(u')\after{m} \eqdef &
 \mbox{if } u'= a \mbox{ then } T\sseq [S_a]  \\
&    \mbox{ else if } u'=u \mbox{ then } T\sseq [S] 
          \mbox{ else } T\sseq \Gamma(u')
\end{array}
$$
To type an input expression the rule {\sc (Type Receive)} requires
the type of $a$ to indicate that the next action is a non marked
matching input action, and every continuation $e_i$ to be well typed
in the type environment where the type of $a$ has advanced to $[S_i]$
and the formal parameters $\vx_i$ have been added. 
Observe that if the input action were marked in the type of $a$, 
it would mean that the input is consumed by the continuation $e_i$,
which would result in a deadlock, as in, e.g., $a\{\rea{m}{a\s m}\}$. 
Finally, the names and the types of the formal parameters are 
recorded in the escape environment of the conclusion of the rule, 
preserving the branching structure of the computation.


According to the rule {\sc (Type End)}, the
expression $\dead$ is well typed assuming that the type of $a$ 
contains no more actions. Moreover, $\Gamma$ must contain no marked
input: a judgement like
$\Gamma,b:[\marked{}?m(T).S]\vdash_a\dead\escape\varnothing$
would mean that the typing of the body of $a$ has assumed to
consume the input of $b$, but it is not the case since the body is
terminated but the action $?m$ of $b$ is still marked. 


\begin{figure}[t]
{\small
\begin{center}
\begin{tabular}{c}
\labeltree{(Type Res Conf)}
{\Gamma,a:T\vdash F\escape\Delta\quad a\notin Dom(\Delta)}
{\gv \New{a} F \escape\Delta, a:T}
\quad
\labeltree{(Type Actor)}
{\gv [a\mapsto M]\quad  \gv_a e\escape\Delta}
{\gv [a\mapsto M]\ a\{e\}\escape\Delta}
\quad
\labeltree{(Type NoMail)}
{\Gamma\vdash\diamond}
{\Gamma\vdash[a\mapsto\varnothing]}
\\ \\
\labeltree{(Type Mailbox)}
{\gv [a\mapsto M] \quad\quad \Gamma\vdash a:[S_a] 
\quad 
?m(\vT).S\in \mathit{Inputs}([S_a]) \Rightarrow 
\vT\sseq \Gamma(\vb)\after{m}  }
{\Gamma\vdash [a\mapsto M\cdot m(\vb)]}
\\ \\
\labeltree{(Type Para)}
{\Gamma_1\vdash F_1\escape\Delta_1 \quad
 \Gamma_2\vdash F_2\escape\Delta_2  \quad
\actors{F_1}\cap\actors{F_2}=\varnothing }
{\Gamma_1,F_1\odot\Gamma_2,F_2\vdash F_1 ~|~F_2
  \escape\Delta_1,\Delta_2}
\quad 
\labeltree{(Type Dead)}
{\Gamma\vdash\diamond\quad \mathsf{NoMark}(\Gamma)}
{\Gamma\vdash \dead\escape\varnothing}
\\ \\
\labeltree{(Type Top Spawn)}
{
\Gamma_1,b:[S_1]\vdash_b e_1 \escape\Delta_1
\quad 
\Gamma_2,b:[S_2]\vdash e_2\escape\Delta_2
\quad b\notin Dom(\Delta_1,\Delta_2)
} 
{\Gamma_1\uplus\Gamma_2\vdash \val b=\act{e_1}; e_2 \escape
  \Delta_1,\Delta_2, b:[S_1\uplus S_2] }
\vspace{-2mm}
\end{tabular}
\end{center}
}
\label{fig:TypeConf}
\caption{Type Rules for Configurations}
\end{figure}

\paragraph{Type rules for Configurations.}
Rule {\sc (Type Res Conf)} shows that when a new name is introduced,
a corresponding type must be guessed. The new name is local to the
configuration $F$, but it is globally collected in the escape 
environment. The rule requires $a$ to be fresh in $\Delta$, but
the derivability of the judgement in hypothesis implies that
$\Gamma,a:T$ is well formed, hence $a\notin Dom(\Gamma)$. 

In order to type an active actor, the rule {\sc (Type Actor)} requires
the (residual) actor body $e$ to comply with the (residual) sequence
of actions in $\Gamma(a)$.
%
%
Moreover, the mailbox $[a\mapsto M]$ contains the list of messages
$M$ that have been received but not handled yet by the actor $a$.
A message in $M$ will be processed by the actor only if the type
$\Gamma(a)$ contains a matching input action. The rule {\sc (Type
  Mailbox)} does not require that every message has a corresponding
handler in $\Gamma(a)$. However, we show in the following that in well
typed systems 
mailboxes only contain messages that will eventually be
handled by the receiving actor.
Let $\mathit{Inputs}([S])$ be the set of top level
input actions contained in $S$.
Then the rule {\sc (Type Mailbox)} states that 
if a message in the mailbox corresponds to one of the
receivable inputs, then the type of the formal parameter is a suffix 
of the type that $\Gamma$ assigns to the actual parameter. The
notation $T\sseq \Gamma(v)\after{m}$ means that 
$\mbox{if } v=a \mbox{ then } T\sseq S  
\mbox{ else } T\sseq \Gamma(v)$, and it is extended to tuples of
types as expected.

The rules {\sc (Type Dead)} and {\sc (Type Top Spawn)} are similar to
the corresponding rules for actors, hence we reserve a final
discussion for the rule {\sc (Type Para)} for parallel composition. 
The rule {\sc (Type Para)} splits the type environment and the escape
environment so that to ensure that the resources consumed by $F_1\para
F_2$ are consumed either by $F_1$ or by $F_2$.
To illustrate, consider 
$$
a:T_a, b:T_b,...\vdash [a\mapsto M]\, a\{e\}\escape\Delta
\quad\quad\quad
a:T'_a, b:T'_b,...\vdash [b\mapsto M']\, b\{e'\}\escape\Delta'
$$
In order to correctly compose the two actors in parallel, the marked
actions in $T_a$, resp $T_b$, must be disjoint form those 
in $T'_a$, resp. $T'_b$. Moreover, since in the typing of an active
actor the behavioural type of the actor can be a suffix
of the initial type of that actor, 
we have that $T_a\sseq T'_a$ and $T'_b\sseq T_b$. Hence the merge-mark
function $\uplus$ must be extended so to compose a sequence with a
subsequence of actions. Let $S'\mergesuff S$ be a partial function 
defined as $S'\uplus S$ plus 
the following two cases, that apply when $S'$ is a proper suffix of
$S$: 
$$
\begin{array}{rlll}
S'\mergesuff \s m(\vT).S & = &  \s m(\vT).(S'\mergesuff S)
& \mbox{ if } S'\sseq S
\\[2mm]
S'\mergesuff \&_{i\in I}\{?m_i(\vT_i).S_i\} & = &
\&_{i\in I\setminus\{j\}}\{?m_i(\vT_i).S_i,\ ?m_j(\vT_j).(S'\mergesuff
S_j)\} & \mbox{ if }  S'\sseq S_j
\end{array}
$$
In particular we let be undefined the case 
$S'\mergesuff \marked{\&}_{i\in
  I}\{\marked{}?m(\vT).S,?m_i(\vT_i).S_i\}$. Indeed, if 
$T_a=[S']$ and $T'_a=[\marked{\&}_{i\in
  I}\{\marked{}?m(T).S,?m_i(T_i).S_i\}]$, it 
means that the actor $b$ sends the message $m$ to $a$ but the
corresponding input handler is not in the body of $a$ anymore. Hence, 
a type with a marked action must be composed 
with a type containing the same non-marked action, so that 
to ensure that the input ``consumed'' by a thread is actually
``offered'' by a parallel thread.
The type environment composition is defined as follows: 
{\small
$$
(\Gamma_1{,}F_1\odot\Gamma_2{,}F_2)\, (u)\eqdef \left\{
\begin{array}{ll}
\Gamma_2(u)\uplus \Gamma_1(u)
   & \mbox{ if } u\notin \actors{F_1}\cup\actors{F_2}       
\\[2mm]
\Gamma_1(u)\mergesuff \Gamma_2(u)
   & \mbox{ if } u\in\actors{F_1} 
\\[2mm]
\Gamma_2(u)\mergesuff \Gamma_1(u)
 & \mbox{ if } u\in\actors{F_2} 
\end{array}
\right.
$$
}
\noindent
where $\actors{F}$ collects the free names of 
active actors in $F$
 (see Appendix~\ref{app:def}). 

\begin{example}\label{ex:pingTyped}
Consider the program $Pr$ in Example~\ref{ex:PingPongPang}. 
We have that $\varnothing\vdash
Pr\escape\{a:T_a^1\uplus T_a^2,\ b:T_b,\ x:T_x,\ y:T_y\}$, which
comes from the following two judgements where $e_a$, resp.
$e_b$, is the body of the actor $a$, resp. $b$:
$$
\begin{array}{ll}
a:T_a^1\vdash_a  e_a\escape\{x:T_x\}
&
a:T_a^2,\ b:T_b\vdash_b e_b\escape\{y:T_y\}
\\[2mm]
%
T_a^1 = [  ?ping(T_x).\, \s pong(T_y).\,?pang. \fineST]
&
T_x=[\marked{}?pong(T_y).\,\s pang. \fineST]
\\[2mm]
T_a^2=[ \marked{}?ping(T_x).\,\s pong(T_y).\,?pang. \fineST]
&
T_y=[\marked{}?pang.\fineST] 
\\[2mm]
T_b  =  [\s ping(T_x).\, ?pong(T_y).\, \s pang.\fineST] 
 &
\end{array}
$$
\end{example}

\paragraph{Preventing deadlocks.}
The type system described so far is enough to prove that actor
implementations comply with the prescribed protocol, however program
execution may stuck in a deadlock state, as for the program
$
P=\val a=\act{\val b=\act{\rea{n}{a\s m}}\, ;\, \rea{m}{b\s n}}
$
which is so that  
$$P\red
[a\mapsto\varnothing]\, a\{\rea{m}{b\s n}\} \para
[b\mapsto\varnothing]\, b\{\rea{n}{a\s m}\}\not\red$$
In order to prevent deadlocks we propose a simple technique that
nicely copes with finite actor computation. We add more structure to
types: we modify
the syntax of types so to have output
actions of the form $T\s m(\vT).S$, where the additional component $T$
describes the sequence of actions performed by the target actor 
\emph{after processing the message $m$}. For instance,
$a:[[S_b]\s m(T).\fineST]$ is the type of an actor $a$ that sends the
message $m$ to an actor that eventually reads the message $m$ and then
continues as described by $S_b$. Let $b$ be the target of such a
message, and let be $b:[S.?m(T).S_b]$. In asynchronous communication, 
when $a$ delivers the message to $b$, it cannot know 
when $b$ will process such a message, but it can safely assume that
after the input of $m$, $b$ will continue as $S_b$.
Adding such a piece of information into types is enough to disallow
deadlocks. Indeed, the typing of the program $P$ above requires
(overlooking the markings) the assumptions
$a:[?m.T'\s n.\fineST], b:[?n.T''\s m.\fineST]$, that are not well
defined since $T'$ and $T''$ can only be mutual recursively defined:
$T'=[T''\s m.\fineST]$ and $T''=[T'\s n.\fineST]$. In other terms,
there are no (finite) types so that $P$ is well typed. 

It turns out that the refinement of types leaves unchanged most of the
type rules presented above. We only have to do a couple of
modifications.
First, the type assumption for the actor $a$ in the rule {\sc (Type
  Send)} must be $\gv a[[S']\s m(\vT).S_a]$, with $S'$ equal to $S$
but for the markings. That is we add to the
output action the sequence $S$ indicated in the type of the target
actor $u$ as the continuation after the input of $m$.
Then we have to adapt the relations between types that we introduced:
%
%
$\uplus$, resp. $\mergesuff$, are obtained
by adapting the clauses dealing with output actions, i.e.,
$T\s m(\vT).S \uplus T\s m(\vT).S'  =  T\s m(\vT).(S\uplus S')$,
resp. $S'\mergesuff T\s m(\vT).S  =   T\s m(\vT).(S'\mergesuff S)$.

\subsection{Properties of the Type System}

We show that the type system respects the semantics of AC,
i.e., well typed configurations reduce to well typed
configurations. However, since actor types 
decrease as long as the computation proceeds,
the subject reduction theorem relies on the following notion of
environment consumption. 

\begin{definition}[Environment Consumption]\mbox{ }
\begin{itemize}
\item
We write $\Gamma'\sseq\Gamma$ when $Dom(\Gamma)=Dom(\Gamma')$ and 
$\forall
u\in Dom(\Gamma)$, $\Gamma'(u)=[S'], \Gamma(u)=[S]$ such that 
$S'\sseq S$. 
\item
$\Delta'\sseq\Delta$ when $\Delta=\escS_{i\in I}\Gamma_i$, 
$\Delta'=\escS_{j\in J}\Gamma'_j$ with $J\subseteq I$,
$\Gamma'_j\subseteq\Gamma_j$ 
and $\forall u\in Dom(\Gamma'_j), \forall j\in J.$ 
$\Gamma'_j(u)\sseq\Gamma_j(u)$. 
\end{itemize}
\end{definition}

The substitution lemma allows a name $b$ to be substituted for a
variable $x$. In the lemma the type of $x$ is assumed to be 
a suffix of the type of $b$, and when $x$ is unified with $b$,
the markings assumed for $b$ must be updated so that they also contain
those assumed for $x$. With an abuse of notation, when $S'\sseq S$ we
let $S'\uplus S$ be defined as $S'\mergesuff S$ plus the clause
$S'\uplus \marked{\&}_{i\in I}\{\marked{}?m(T).S,?m_i(T_i).S_i\}  = 
\marked{\&}_{i\in I}\{\marked{}?m(T).(S'\uplus S),?m_i(T_i).S_i,\}$,
that were forbidden in the composition of parallel threads.
Such a clause here is not a problem since the substitution lemma 
applies within a single thread, i.e. within the body of an actor at
the moment of receiving an input, where it is safe to merge local
assumptions.

\begin{lemma}[Substitution]
Let be $\Gamma, x:T\vdash_a e\escape \Delta$
\begin{itemize} 
\item let $c$ be an actor name s.t. $a\neq c$ and 
 $T\sseq\Gamma(c)$, then 
 $\Gamma; c:T\uplus\Gamma(c)\vdash_ae\{^c/_x\}\escape\Delta$; 
\item let be $?m(T).S\in Input(\Gamma(a))$ such that $T\sseq [S]$,
then $\Gamma; a:T\uplus\Gamma(a)\vdash_a e\{^a/_x\}\escape\Delta$. 
\end{itemize}
\end{lemma}

\begin{theorem}[Subject Reduction]\mbox{ }
If $\Gamma\vdash F \escape\Delta$ and $F \red F'$, 
then there exist $\Gamma'$ such that 
$\Gamma'\vdash F'\escape\Delta'$, with 
$\Gamma'\sseq\Gamma$ and $\Delta'\sseq\Delta$. 
\end{theorem}

Let $F$ be a well typed closed system, i.e. 
$\varnothing\vdash F\escape\Delta$.
We have that any input action that is marked in $\Delta$ exactly
corresponds to one output expression in $F$ that eventually consumes
that input. We say that a well typed actor system is \emph{balanced}
whenever \emph{every} input in the system appears marked in 
the escape environment. As a consequence, in balanced systems every
input has a matching output and viceversa. Then 
(finite) balanced systems eventually terminate in the empty
configuration, correctly implementing the communication
protocol defined by the typing.


The definition of balanced environment checks that
every actor has a fully marked type, possibly with the
contribution of the markings contained in the types of a number of
variables. Indeed, since actor names are passed as parameters, the
inputs offered by an actor $a$ can be consumed by outputs directed to
variables that are dynamically substituted with $a$, as in
$b\{\rea{foo(x)}{x\s m}\}\para a\{b\s foo(a);\rea{m}{e}\}$.
Let be $\fullym{[S]}=[S']$ where $S'$ and $S$ are the same
sequence of actions, but in $S'$ every top level input is marked.

\begin{definition}[Balanced environment]
We write $\mathsf{balanced}(\Delta)$ when
\begin{itemize}
\item if $\Delta=\{x_1:T_1,\ldots,x_n:T_n\}$ then 
  $\mathsf{NoMark}(T_1),\ldots \mathsf{NoMark}(T_n)$; 
\item if $\Delta=\{u_1:T_1,\ldots,u_n:T_n\}$, then 
for any name $a\in Dom(\Delta)$ 
\begin{enumerate}
\item $\exists\ x_1,...,x_k\in Dom(\Delta)$ such that $\Delta(a)=T$,
  $\Delta(x_i)=T_i$ with $T_i\sseq T$ and 
  $((T\uplus T_1)\uplus\ldots\uplus T_k)=\fullym{T}$
\item
  $\mathsf{balanced}(\Delta\setminus\{a:T, x_1:T_1,...,x_k:T_k\})$;
\end{enumerate}
\item if $\Delta=\escS_{i\in I}\Gamma_i$, then
  $\mathsf{balanced}(\Gamma_i)$ for any $i\in I$.
\end{itemize}
\end{definition}


\noindent
Observe that the escape environment in
Example~\ref{ex:pingTyped} is balanced. Indeed we have that
$\Delta=\{a{:}T_a^1{\uplus} T_a^2,\ b{:}T_b,\ x{:}T_x,\ y{:}T_y\}$, 
$\Delta(y)\sseq \Delta(a)$ and 
$\Delta(y)\uplus \Delta(a)=\fullym{\Delta(a)}$.
Similarly, $\Delta(x)\sseq \Delta(b)$ and  
  $\Delta(x)\uplus \Delta(b)=\fullym{\Delta(b)}$.

Let $\red^*$ be the transitive closure of $\red$.
A final lemma shows that during the computation of well typed
actor systems mailboxes only contain messages that are eventually 
handled by the receiving actor.

\begin{lemma}\label{lem:mailbox}
If $\varnothing\vdash Pr\escape \Delta$ 
and $Pr\red^* \New{a_1,..,a_k}([a\mapsto M]a\{e\}\para F)$, 
then there exists 
$\Gamma$ such that $\Gamma\vdash [a\mapsto M]a\{e\}$ 
and for any $m(v)\in M$, there exists a matching input action
$?m(T).S$ that belongs to $Inputs(\Gamma(a))$.
\end{lemma}

\begin{theorem}[Safety]
If $\varnothing\vdash Pr\escape\Delta$ with  $\balanced(\Delta)$.
If $Pr\red^* F$ then either $F=\dead$ or $F\red
F'$ for some $F'$.
\end{theorem}

\begin{example}
Consider the program $\hat{Pr}$:
$$
\begin{array}{l}
\val a=\act{\rea{ping(x)}{x\s pong(\self);\rea{pang()}{\dead}}}\ ;
\\[2mm]
\val b=\act{a\s ping(\self);\rea{pong(y)}{\dead}}\ ;\ \dead
\end{array}
$$
It is easy to see that $\hat{Pr}\red^*
[a\mapsto\varnothing]\ a\{\rea{pang()}{\dead}\}\para \dead\not\red$
since there is no actor sending the message that $a$ is waiting for.
Nevertheless, the program can be typed, but with an escape environment
that is not balanced. Indeed,  
$\varnothing \vdash \hat{Pr}\escape \Delta$ is derivable with
$\Delta= \{ 
a:[\marked{}?ping(T_x).T^*]$,
$b:[T^*\s ping(T_x).?pong([?pang.\fineST]).\fineST]$,
 $x:T_x{=}[\marked{}?pong([?pang.\fineST]).\fineST],\
 y:[?pang.\fineST] \ \}
$
where 
$T^*=[ [\fineST]\s pong([?pang.\fineST]).?pang.\fineST]$.
Then the communications are well typed but the fact that in $\Delta$
there is no mark for the input $?pang$ shows that the program does not
consume that resource, as indeed the operational semantics above has
shown. 
\end{example}

\begin{example}
As a final example observe that the deadlock program discussed above,
that is 
$
P=\val a=\act{\val b=\act{\rea{n}{a\s m}}\, ;\, \rea{m}{b\s n}}
$, is balanced since $a$'s input is matched by $b$'s output and
viceversa. However, the program stucks in a deadlock and indeed it is
not well typed, according to the safety theorem.
\end{example}

\section{Conclusions and Related Work}
\label{sec:conclude}

We presented AC, a core actor calculus designed around
the basic primitives of the Scala's Actors API, together with 
a behavioural type system and a safety property dealing with the
determinism of finite actor communications. 
We think that this work sheds light on how formal methods developed in
the context of linear types and session types for the
$\pi$-calculus can be profitably reused for the analysis of actor
systems. 

As we pointed out in the Introduction, our type system draws
inspiration from the formal methods developed in the contexts of 
linear types and session types for process algebras
(\cite{THK94,HVK98,KPT99,GayV10}).  
Indeed, the Actor programming model shares with session types the idea
of conceiving the computation as the implementation of a specified
communication protocol. However, there are a number of key differences
between the two models. First, in asynchronous session types
(\cite{GayV10,CDYAsynch07}) two sequential outputs are
(asynchronously) processed according to the sending order, while if an
actor $a$ sends two messages to the actor $b$, i.e. $a\{b\s m_1;b\s
m_2\}$, then $b$ is free to read/process the second message before
reading/processing the first one,
i.e. $b\{\rea{m_2}{\rea{m_1}{...}}\}$. This means 
that in the Actor model we have to deal with looser assumptions, in
that we cannot look at the order of input, resp. output, actions to
infer something about the order of the dual output, resp. input
actions. 
More importantly, in multiparty session types (\cite{HYC08,DY11}) the
set of interacting parties (or the set of interacting roles) in a
given session is known from the beginning, while an actor system is a
dynamic set of interacting parties. In particular, since new actors
can be created and actor names can be passed as parameters, the
communication capabilities of actors dynamically increase,
in a way similar to the scope extrusion phenomenon of the
$\pi$-calculus. 

Intuitively, session types describe the flow of
communications withing a single conversation session. An actor's
behavioural type instead takes the point of view of an entity that
concurrently participates to different (interleaved) conversations
with different parties. In this sense the Actor model share some
similarities with the Conversation
Calculus~\cite{ConversCalc,ConversTypes} CC, where 
processes concurrently participate to multiparty conversations and
conversation context identities can be passed around to allow
participants to dynamically join conversations. The Conversation
Calculus is designed to model service-oriented computation, and it is
centered around the notion of conversation context, which is a medium
where related interactions take place. The main difference with the
Actor model is that in CC named entities are the conversation
contexts, while named actors are the conversant parties. Hence the
powerful type system in~\cite{ConversTypes} associates types to
conversations rather than to actors.

To he best of our knowledge, there are few works dealing with 
type systems for Actor calculi. In \cite{AghaT04} the Actor model is
encodes in a typed variant of the $\pi$-calculus, where types are used
to ensure uniqueness of actor names and freshness of names of newly
created actors. 
The work in~\cite{ColacoPDS99} study a type system for a primitive
actor calculus, called CAP, which is essentially a calculus of
concurrent objects \`a la Abadi and Cardelli~\cite{AC94} where actors
are objects that dynamically, that is in response to method
invocation, change the set of available methods. Such a dynamic
behaviour may lead to so called ``orphan messages'' which may not be
handled by the the target actor in some execution path. In order to
avoid such orphan messages, a type system is proposed so to provide a
safe abstraction of the execution branches.
Finally, In~\cite{FeatErlang} a concurrent fragment of the Erlang
language is enriched with sessions and session types. The safety
property guaranteed by the typing is that all within-session messages
have a chance of being received and sending and receiving follows the
patterns prescribed by types. In our work we followed a
different approach: instead of adding sessions to an actor calculus,
we reused session type techniques to deal with the communication model
distinctive of Actors.

As for future work we plan to extend the AC calculus to deal with
recursive actors. Infinite computation requires a different
formulation 
of the safety property, since compliance with the intended protocol
does not reduces anymore to the termination of all actors with empty
mailboxes. Moreover, deadlock freedom requires more sophisticated
techniques, such as those in \cite{CDYAsynch07,ConversTypes} that are
based on a proof system that identifies cyclic dependencies between
actions.




\vspace{2mm}
{\bf Acknowledgements.} The author is indebted to Mariangiola
Dezani-Ciancaglini and Luca Padovani for insightful discussions about
session type theories.

\bibliographystyle{splncs03}
\bibliography{scala}

\appendix
\section{Notation and Useful Definitions}
\label{app:def}

Let $\&_{i\in I}^{\circ}\{?m_i(T_i).S_i\}$ stands for an input
action, either marked or not.

\begin{definition}[Suffix]
Let $S$ and $S'$ be two sequences of actions, we write $S\sseq S'$
when $S$ is a suffix of $S'$ according to the following rules:  
$$
\begin{array}{l}
\infer{\fineST\sseq S}{}
\quad\quad\quad
\infer{S\sseq S}{}
\quad\quad\quad
\infer{S\sseq T\s m(T').S'}{S\sseq S'}
\quad\quad\quad
\infer{S\sseq \&_{i\in I}^{\circ}\{?m_i(T_i).S_i\}}{S\sseq S_i\ \ \exists i\in I}
\end{array}
$$
\end{definition}

\begin{definition}[Inputs] The set of top level inputs contained in a
  type $T$, written $Inputs(T)$, is defined as follows:
$$
\begin{array}{l}
Inputs([T\s m(T').S])=Inputs([S])
\\[2mm]
Inputs([\&^\circ_{i\in
  I}\{^\circ?m_i(T_i).S_i\}])=\bigcup_{i\in I}\{?m_i(T_i).S_i\}\cup
Inputs([S_i])\\[2mm]
Inputs([\fineST])=\varnothing
\end{array}
$$
\end{definition}

\begin{definition}[Free active actors]
The set $\actors{F}$ of free active actors in the configuration
$F$ is defined as follows:
$$
\begin{array}{ll}
\actors{\dead}=\actors{e}=\varnothing
&
\actors{[a\mapsto M]a\{e\}}=\{a\}
\\[2mm]
\actors{F_1\para F_2}=\actors{F_1}\cup\actors{F_2}\ \ \ 
&
\actors{\New{a}{F}}=\actors{F}\setminus\{a\}
\end{array}
$$
\end{definition}

\section{Proof Sketches}



\noindent
{\bf Substitution Lemma} 
Let be $\Gamma, x:T\vdash_a e\escape \Delta$ with $x\neq a$,
\begin{itemize} 
\item let $c$ be an actor name s.t. $a\neq c$ and 
 $T\sseq\Gamma(c)$, then 
 $\Gamma; c:T\uplus\Gamma(c)\vdash_ae\{^c/_x\}\escape\Delta$; 
\item let be $?m(T).S\in Input(\Gamma(a))$ such that $T\sseq [S]$,
then $\Gamma; a:T\uplus\Gamma(a)\vdash_a e\{^a/_x\}\escape\Delta$. 
\end{itemize}
\begin{proof}
The proof is by induction on the derivation of the judgement 
 $\Gamma, x:T\vdash_a e\escape \Delta$. The base case is when
 $e=\dead$. In this case the hypothesis is  $\Gamma, x:T\vdash_a
 \dead\escape \varnothing$, $\Gamma,x:T\vdash a:[\fineST]$ and
$\mathsf{NoMark}(\Gamma,x:T)$, and the thesis is $\Gamma;
c:T\uplus\Gamma(c)\vdash_a\dead\escape\varnothing$. Hence it is
sufficient to show that 
$\Gamma; c:T\uplus\Gamma(c)\vdash a:[\fineST]$, which is immediate if
$a\neq c$. On the other hand, if $a=c$ then from the first hypothesis
we have $\Gamma(a)=[\fineST]$ and from the second one we have
$T\sseq[\fineST]$, hence $T=[\fineST]=T\uplus[\fineST]$ as desired.
For the inductive cases we proceed by a case analysis on the last rule
that has been used to derive $\Gamma, x:T\vdash_a e\escape \Delta$:
\begin{description}
\item[({\sc Type Spawn})] In this case the hypothesis are
$\Gamma_1,x:T_1\uplus\Gamma_2,x:T_2\vdash_a 
\val b=\act{e_1}; e_2 \escape
  \Delta_1,\Delta_2,\{b:[S_1\uplus S_2]\} $
and $T=T_1\uplus T_2 \sseq (\Gamma_1\uplus\Gamma_2)(c)$. The first
judgement comes from
$\Gamma_1,x:T_1,b:[S_1]\vdash_c e_1 \escape\Delta_1$ and
$\Gamma_2,x:T_2,b:[S_2]\vdash_a e_2\escape\Delta_2$. Observe that $b$
is fresh, then $b\neq c$, then by induction
we have $\Gamma_1;c:T_1\uplus\Gamma_1(c),b:[S_1]\vdash_b
e_1\{^c/_x\} \escape\Delta_1$ and 
$\Gamma_2;c:T_2\uplus\Gamma_2(c),b:[S_2]\vdash_a
e_2\{^c/_x\}\escape\Delta_2$. Then by {\sc (Type Spawn)} we have 
$\Gamma'\vdash_a\val b=\act{e_1\{^c/_x\}}; e_2\{^c/_x\} \escape
  \Delta_1,\Delta_2,\{b:[S_1\uplus S_2]\}$, where 
$\Gamma'=\Gamma_1;c:T_1\uplus\Gamma_1(c)\, \uplus\,
\Gamma_2;c:T_2\uplus\Gamma_2(c)$. Now observe that 
$\Gamma'= (\Gamma_1\uplus\Gamma_2);
c:(T_1\uplus\Gamma_1(c))\uplus(T_2\uplus\Gamma_2(c))= 
(\Gamma_1\uplus\Gamma_2);c:(T\uplus(\Gamma_1\uplus\Gamma_2)(c))$ as
desired. 
\item[{\sc (Type Send)}] For simplicity let assume a single message
  parameter. In this case the hypothesis is 
$\Gamma,x:T\vdash_a u\s m(u'); e \escape \Delta$, 
which come from 
     \begin{enumerate}
     \item $\Gamma,x:T\vdash a: [[S]\s m(T').S_a]$;
     \item $\Gamma,x:T\vdash u: [S_u\,.\,\marked{\&}_{i\in
         I}\{\marked{}?m(T').S,\ ?m_i(\vT_i).S_i\}]$;
     \item $T'\sseq (\Gamma,x:T)(u')\after{m}$, that is  
       if  $u'= a$ then $T'\sseq [S_a]$ else if $u'=u$ then $T'\sseq
       [S]$ else $T'\sseq \Gamma,x:T(u')$ and
     \item $\Gamma,x:T\,;\, u: [S_u.\&_{i\in I}\{?m(T').S,\
       ?m_i(\vT_i).S_i\}] \,;\, a:[S_a]\vdash_a
       e\escape{\Delta}$. Notice that $x\neq a$, then we have two
       subcases:
       \begin{enumerate}
         \item $x\neq u$, then $\Gamma;\, u: [S_u.\&_{i\in
        I}\{?m(T').S,\ ?m_i(\vT_i).S_i\}] \,;\, a:[S_a],x:T\vdash_a
       e\escape{\Delta}$
         \item $x=u$ and $u\neq a$, then $\Gamma;\,
           a:[S_a],x:[S_u.\&_{i\in I}\{?m(T').S,\ ?m_i(\vT_i).S_i\}]
           \vdash_a  e\escape{\Delta}$. 
       \end{enumerate}
     \end{enumerate}
Now, since $x\neq a$, but it might be the case that $x=u$
and/or $x=u'$. Then we have to prove that
$\Gamma;c:T\uplus\Gamma(c)\vdash_a u\{^c/_x\}\s m(u'\{^c/_x\});
e\{^c/_x\} \escape \Delta$, which comes from the following judgements,
implied by the enumeration above:   
\begin{itemize}
     \item $\Gamma\vdash a: [[S]\s m(T').S_a]$, since $x\neq a$, hence
          $\Gamma;c:T\uplus\Gamma(c)\vdash a: T_a$, where
          \begin{itemize} 
          \item if $a\neq c$ then $T_a= [[S]\s m(T').S_a]$ 
          \item if $a=c$, then $T_a=T\uplus[[S]\s m(T').S_a]$ together
           with the hypothesis $T\sseq [S]$
           \end{itemize} 
     \item $\Gamma;c:T\uplus\Gamma(c)\vdash u\{^c/_x\}: T_u$, where
         \begin{itemize} 
          \item if $u\neq c,x$, i.e. $u\{^c/_x\}=u$, then 
           $T_u=[S_u\,.\,\marked{\&}_{i\in I}\{\marked{}?m(T').S,\
              ?m_i(\vT_i).S_i\}]$
          \item  if $u=c$ or $u=x$, i.e., $u\{^c/_x\}=c$, then 
              $T_u=T\uplus [S_u\,.\,\marked{\&}_{i\in
                I}\{\marked{}?m(T').S,\  ?m_i(\vT_i).S_i\}]$,
             together with the hypothesis 
           $T\sseq\Gamma(c)$ and $T\sseq [S]$ for $c=a$;
          \end{itemize}     
     \item by induction we have
       \begin{itemize}
        \item in the case 4.(a), i.e. $x\neq u,a$,
       $(\Gamma; u: T'_u;
       a:[S_a]);c:T\uplus(\Gamma;u:T'_u;a:[S_a])(c)\vdash_a
       e\{^c/_x\}\escape{\Delta}$ where  
       $T'_u= [S_u.\&_{i\in I}\{?m(T').S,\ ?m_i(\vT_i).S_i\}]$.
       \item in the case 4.(b), i.e. $x=u\neq a$,
       $(\Gamma; a:[S_a]);c:T\uplus(\Gamma;a:[S_a])(c)\vdash_a
       e\{^c/_x\}\escape{\Delta}$ with $T= [S_u.\&_{i\in
         I}\{?m(T').S,\ ?m_i(\vT_i).S_i\}]$ 
       \end{itemize}
     In both cases the environment is equal to \\
      $(\Gamma;c:T\uplus\Gamma(c)); u: [S_u.\&_{i\in I}\{?m(T').S,\ ?m_i(\vT_i).S_i\}]; a:[S_a]$
    \item let prove that  
      $T'\sseq (\Gamma;c:T\uplus\Gamma(c))(u'\{^c/_x\})\after{m}$:
      \begin{itemize}
        \item if $u'\{^c/_x\}= a$ then either $u'=a$, then 
          $T'\sseq [S_a]$  by the hypothesis 3. above, 
          or $u'=x$ and $a=c$, and in this case the hypothesis 3.
          above gives $T'\sseq \Gamma,x:T(u')$, that is $T'\sseq T$. 
          Moreover, from the hypothesis $T\sseq [S]$ and the fact that
          $S\sseq S_a$, we have $T'\sseq [S_a]$.
        \item if $u'\{^c/_x\}= u\{^c/_x\}$ then either $u'=u$, then 
          $T'\sseq [S]$  by the hypothesis 3. above, or $u'\{^c/_x\}=
          u\{^c/_x\}=c$ and as in the previous item $T'\sseq T$ and by
          hypothesis $T\sseq [S]$, hence $T'\sseq [S]$.
       \item  otherwise we have to show that 
          $T'\sseq (\Gamma;c:T\uplus\Gamma(c))(u'\{^c/_x\})$, that is
          either $T'\sseq (\Gamma;c:T\uplus\Gamma(c))(c)$ or $T'\sseq
          (\Gamma;c:T\uplus\Gamma(c))(u')$.
          The first case come from the fact that by hypothesis
          3. above we have $T'\sseq T$, 
          that together with $T\sseq\Gamma(c)$ gives what desired.
          The sencond case come from the fact that by hypothesis
          3. above we have $T'\sseq \Gamma(u')$, which is what desired
          since in this case $u'\neq c$.
     \end{itemize}
    \end{itemize}
\item[{\sc (Type Receive)}] In this case the hypothesis is 
$\Gamma,x:T\vdash_a \rea{m_i(\vy_i)}{e_i}_{i\in I}\escape \escS_{i\in
  I} \, (\Delta_i,\{\vy_i:T_i\})$, which comes from
$(i)~\Gamma,x:T\vdash a: [ \&_{i\in I} \{?m_i(\vT_i).S_i\}]$ and 
$(ii)~\Gamma,x:T\, ;\, a:[S_i], \vy_i:\vT_i\vdash_a e_i \escape\Delta_i$
for $i\in I$. We can assume that $x\neq \vy_i$, hence from the last
judgement we have $\Gamma;a:[S_i], \vy_i:\vT_i,x:T\vdash_a e_i
\escape\Delta_i$, which gives by inductive hypothesis
$\Gamma;a:[S_i],\vy_i:\vT_i;c:T\uplus(\Gamma;a:[S_i])(c)\vdash_a
e_i\{^c/_x\} \escape\Delta_i$, that is
$\Gamma;c:T\uplus(\Gamma)(c);a:[S'_i],\vy_i{:}\vT_i\vdash_a
e_i\{^c/_x\} \escape\Delta_i$ where $S'_i=S_i\uplus T$ if $c=a$,
otherwise $S'_i=S_i$. Now, from $(i)$ and the hypothesis $T\sseq
[S_i]$ of the lemma, we have
$\Gamma;c:T\uplus(\Gamma)(c)\vdash a: [ \&_{i\in I}
\{?m_i(\vT_i).S'_i\}]$,hence we conclude
$\Gamma;c:T\uplus(\Gamma)(c)\vdash_a
\rea{m_i(\vy_i)}{e_i\{^c/_x\}}_{i\in I}\escape \escS_{i\in I} \,
(\Delta_i,\{\vy_i:T_i\})$ by {\sc (Type Receive)}.
\end{description}

\end{proof}

\medskip\noindent
\begin{lemma}\label{lem:odot}
Let be $\Gamma=\Gamma_1,F_1\odot\Gamma_2,F_2\vdash F_1\para
F_2\escape\Delta_1,\Delta_2$ a derivable judgement. Then 
\begin{itemize}
\item For all $u\notin\actors{F_1}\cup\actors{F_2}$,
  $\Gamma(u)=\Gamma_1(u)\uplus\Gamma_2(u)$, hence 
$\Gamma_1(u)\sseq\Gamma_2(u)$ and $\Gamma_1(u)\sseq\Gamma_2(u)$,
i.e. the type of $u$ has the same length in $\Gamma_1$ and $\Gamma_2$
\item For all $u\in\actors{F_1}$, it holds
$\Gamma_1(u)\sseq\Gamma_2(u)$
\item For all $u\in\actors{F_2}$, it holds
$\Gamma_2(u)\sseq\Gamma_1(u)$
\end{itemize}

\end{lemma}

\medskip\noindent
{\bf Subject Reduction Theorem}
If $\Gamma\vdash F \escape\Delta$ and $F \red F'$, 
then there exist $\Gamma'$ such that 
$\Gamma'\vdash F'\escape\Delta'$, with 
$\Gamma'\sseq\Gamma$ and $\Delta'\sseq\Delta$. 
%
\begin{proof}
The proof is by induction on the derivation of $F\red F'$. We start
with the base cases:
\begin{description}
\item[{\sc (Ended)}] in this case the hypothesis are $\gv
[a\mapsto\varnothing]\, a\{\dead\}\escape\Delta$ and
$[a\mapsto\varnothing]\, a\{\dead\}\red \dead$. From the first
hypothesis we have that $\gv [a\mapsto\varnothing]$ and
$\gv_a \dead\escape\Delta$, hence $\Delta=\varnothing$, $\gv
a:[\fineST]$ and $\mathsf{NoMark}(\Gamma)$. Then $\gv\diamond$
is also derivable, and by {\sc (Type Dead)} we conclude
$\gv\dead\escape\Delta$ as desired. 
\item[{\sc (Top Spawn)}] In this case the hypothesis are 
$\val a=\act{e};e' 
 \red   \New{a}([a\mapsto\varnothing]\, a\{e\}\}\para e'))$ and
$\gv \val a=\act{e};e'\escape\Delta$, which comes form
$\Gamma=\Gamma_1\uplus\Gamma_2$,
$\Delta=\Delta_1,\Delta_2,a:[S_1\uplus S_2]$, 
$\Gamma_1,a:[S_1]\vdash_a e \escape\Delta_1$ and
$\Gamma_2,a:[S_2]\vdash e'\escape\Delta_2$. Then we also have 
$\Gamma_1, a:[S_1]\vdash [a\mapsto\varnothing]\, a\{e\}\escape
\Delta_1$, and by {\sc (Type Para)},
$\Gamma'\vdash [a\mapsto\varnothing]\, a\{e\}\para e'\escape
\Delta_1,\Delta_2$ where
$\Gamma'=\Gamma_1,a:[S_1]\odot\Gamma_2,a:[S_2]=
\Gamma_1\uplus\Gamma_2,a:[S_1\uplus S_2]$.
Then by {\sc (Type Res Conf)} we conclude
$\Gamma\vdash \New{a}([a\mapsto\varnothing]\,
a\{e\}\para e')\escape \Delta$. 
\item[{\sc (Spawn)}] in this case the hypothesis are 
$[b\mapsto M]b\{ \val a=\act{e};e'\} 
 \red   \New{a}([b\mapsto M]b\{ e'\}\para [a\mapsto\varnothing]\,
 a\{e\})$ and 
$\gv [b\mapsto M]b\{ \val a=\act{e};e'\} \escape\Delta$, which comes
form 
$\gv [b\mapsto M]$ and
$\gv_b \val a=\act{e};e'\escape\Delta$. Then the proof is similar to
the previous case.
\item[{\sc (Send)}] In this case the hypothesis are 
$[a\mapsto M]\, a\{e\}\,\para\, [b\mapsto M']\, b\{a\s m(\vc);e'\}
\red  
[a\mapsto M\cdot m(\vc)]\, a\{e\}\,\para\, [b\mapsto M']\,b\{e'\}$
and
$\gv [a\mapsto M]\, a\{e\}\,\para\, [b\mapsto M']\, b\{a\s
m(\vc);e'\}\escape\Delta$, which comes from
$\Gamma=\Gamma_1,F_1\odot\Gamma_2,F_2$, $\Delta=\Delta_1,\Delta_2$,
$\Gamma_1\vdash [a\mapsto M]\, a\{e\}\escape\Delta_1$ and
$\Gamma_2\vdash [b\mapsto M']\, b\{a\s m(\vc);e'\}\escape\Delta_2$.
The last two judgements must have been derived from
$(i)~\Gamma_1\vdash [a\mapsto M]$ and 
$(ii)~\Gamma_1\vdash_a e\escape\Delta_1$,
resp. $(iii)~\Gamma_2\vdash [b\mapsto M']$ and 
$(iv)~\Gamma_2\vdash_b a\s m(\vc);e'\escape\Delta_2$. 

From $(iv)$ he know that $\Gamma_2\vdash b: [[S]\s m(\vT).S_b]$,
$\Gamma_2\vdash a: [S_u\,.\,\marked{\&}_{i\in I}\{\marked{}?m(\vT).S,\
?m_i(\vT_i).S_i\}]$, $(*)~\vT\sseq \Gamma_2(\vc)\after{m}$ and
$\Gamma_2'\vdash_b e'\escape{\Delta_2}$, then also
$\Gamma'_2\vdash b\{e'\}\escape \Delta_2$, where 
$\Gamma'_2=\Gamma_2\,;\, a: [S_u.\&_{i\in I}\{?m(\vT).S,\
?m_i(\vT_i).S_i\}] \,;\,  b:[S_b]$. Notice that
$\Gamma'_2\sseq\Gamma_2$. Moreover, from $(iii)$ we also have
$\Gamma'_2\vdash [b\mapsto M']$ since
$Inputs(\Gamma'_2(b))=Inputs(\Gamma_2(b))$. Then by {\sc (Type Actor)}
we have $\Gamma'_2\vdash [b\mapsto M']\, b\{e'\}\escape\Delta_2$.

Let show that from $(i)$ we also have $\Gamma_1\vdash [a\mapsto M\cdot
m(\vc)]$ by {\sc (Type Mailbox)}. From  
$\Gamma_2\vdash a:[S_u\,.\,\marked{\&}_{i\in I}\{\marked{}?m(\vT).S,\
?m_i(\vT_i).S_i\}]$ we have $?m(\vT).S\in Inputs(\Gamma_2(a))$, then
by Lemma~\ref{lem:odot} $?m(\vT).S\in Inputs(\Gamma_1(a))$. It is then
sufficient to show that $\vT\sseq\Gamma_1(\vc)\after{m}$, that is 
if $c=a$ then $T\sseq S$ else $T\sseq\Gamma_1(c)$. 
From $(*)$ above we know that 
if $c=b$ then $T\sseq\Gamma_2(b)$, if $c=a$ then $T\sseq S$ else
$T\sseq\Gamma_2(c)$. By Lemma~\ref{lem:odot}
$\Gamma_2(c)\sseq\Gamma_1(c) $ and $\Gamma_2(b)\sseq\Gamma_1(b)$ hence
we have $\vT\sseq\Gamma_1(\vc)\after{m}$ as desired.

So we have $\Gamma_1\vdash [a\mapsto M\cdot m(\vc)]$, that together
with $(ii)$ gives $\Gamma_1\vdash [a\mapsto M\cdot m(\vc)]\,
a\{e\}\escape\Delta_1$. Then by {\sc (Type Para)} we have
$\Gamma'\vdash [a\mapsto M\cdot m(\vc)]\,a\{e\}\para [b\mapsto M']\,
b\{e'\}\escape\Delta_1,\Delta_2$ where
$\Gamma'=\Gamma_1\odot\Gamma'_2$, i.e. $\Gamma'\sseq\Gamma$ as
desired. 
\item[{\sc (Receive)}] In this case the hypothesis are
$[a\mapsto M\!\cdot\! m_j(\vc)\!\cdot\! M']\ 
 a\{ \rea{m_i(\vx_i)}{e_i}_{i\in I}\} \red  
  [a\mapsto M\!\cdot\! M']\ a\{e_j\{^\vc/_{\vx_j}\} \}$ with $j\in I$ 
and
$\gv [a\mapsto M\!\cdot\! m_j(\vc)\!\cdot\! M']\ 
 a\{ \rea{m_i(\vx_i)}{e_i}_{i\in I}\} \escape\Delta$, which comes form
$(i)~\gv [a\mapsto M\!\cdot\! m_j(\vc)\!\cdot\! M']$ and 
$(ii)~\gv_a  \rea{m_i(\vx_i)}{e_i}_{i\in I}\escape\Delta$, where
$\Delta=\escS_{i\in I} \, (\Delta_i,\{\vx_i:T_i\})$. From $(ii)$ and
$j\in I$ we have $\gv a: [ \&_{i\in I} \{?m_i(\vT_i).S_i\}]$ and
$\Gamma\, ;\, a:[S_j], \vx_j:\vT_j\vdash_a e_j\escape\Delta_j$.
Now, from $(i)$ and $?m_j(\vT_j).S_j\in Input(\Gamma(a))$ we have
$\vT\sseq\Gamma(\vc)\after{m}$, that is if $c=a$ then $T\sseq S_j$  
else $T\sseq \Gamma(c)$.

Let be $\Gamma'=\Gamma;a:[S_j]$, then we also have
$\vT\sseq\Gamma'(\vc)\after{m}$. 
From $\Gamma', \vx_j:\vT_j\vdash_a e_j\escape\Delta_j$
and $\vT\sseq\Gamma'(\vc)\after{m}$, by Substitution Lemma we have
$\Gamma'; \vc:\vT_j\uplus\Gamma'(\vc)\vdash_a
e_j\{^{\vc}/_{\vx}\}\escape\Delta_j$. Let be $\Gamma^*=\Gamma';
\vc:\vT_j\uplus\Gamma'(\vc)$. Note that $\Gamma^*\sseq\Gamma$ and from
$(i)$ we also have $\Gamma^*\vdash[a\mapsto M\cdot M']$, than by {\sc
  (Type Actor)} we conclude $\Gamma^*\vdash [a\mapsto M\cdot M']\vdash
a\{e_j\{^{\vc}/_{\vx}\}\}\escape\Delta_j$, with $\Delta_j\sseq\Delta$.
\end{description}

\noindent For the inductive cases we proceed by a case analysis on the
last rule tha has been applied:
\begin{description}
\item[{\sc (Par)}] In this case the hypothesis is $F_1\para F_2 \red
  F'_1\para F_2$ since $F_1\red F'_1$, and $\gv F_1\para
  F_2\escape\Delta$. Hence $\Delta=\Delta_1,\Delta_2$, 
$\Gamma=\Gamma_1,F_1\odot\Gamma_2,F_2$, $\Gamma_1\vdash
F_1\escape\Delta_1$ and $\Gamma_2\vdash F_2\escape\Delta_2$. Then by
inductive hypothesis we have $\Gamma'_1\vdash F'_1\escape\Delta'_1$
with $\Gamma'_1\sseq\Gamma_1$ and $\Delta'_1\sseq\Delta_1$. Then by
{\sc (Type Para)} we have $\Gamma'\vdash F'_1\para F_2\escape\Delta'$
with $\Delta'=\Delta'_1,\Delta_2$ and
$\Gamma'=\Gamma'_1,F'_1\odot\Gamma_2,F_2$ (note that we can guarantee
that $\actors{F'_1}\cap \actors{F_2}=\varnothing$ since we assumed
  that dynamically spawned actors have fresh names). Then observe that
  $\Delta'\sseq\Delta$, hence it is sufficient to show that 
$\Gamma'\sseq\Gamma$, which comes oberving that
$\Gamma'_1(u)=\Gamma_1(u)$ for all $u\notin\actors{F_1}$ while
$\Gamma'_1(u)\sseq\Gamma_1(u)$ for all
$u\in\actors{F_1}\cap\actors{F'_1}$. 

\item[{\sc (Res)}]  In this case the hypothesis is $\New{a}F \red
  \New{a}F'$ since $F\red F'$, and $\gv \New{a}F\escape\Delta,a:T$.
 The last judgement comes from $\Gamma,a:T\vdash F\escape \Delta$,
 which gives, by inductive hypothesis, $\Gamma',a:T'\vdash
 F'\escape\Delta'$ with $\Gamma',a:T'\sseq\Gamma,a:T$ and
 $\Delta'\sseq\Delta$, and we conclude 
$\gv \New{a}F'\escape\Delta,a:T'$ by {\sc (Type Res Conf)}.
\item[{\sc (Struct)}] In this case the hypothesis is $F \red F'$ 
 since $F\equiv F'$, $F'\red F''$ and $F''\equiv F'''$. This case
 comes by the fact that structural congruence preservs the typing. 
\end{description}

\end{proof}

\medskip\noindent
{\bf Lemma \ref{lem:mailbox}}
If $\varnothing\vdash Pr\escape \Delta$ 
and $Pr\red^* \New{\va}([a\mapsto M]a\{e\}\para F)$, 
then there exists 
$\Gamma$ such that $\Gamma\vdash [a\mapsto M]a\{e\}$ 
and for any $m(v)\in M$, there exists a matching input action
$?m(T).S$ that belongs to $Inputs(\Gamma(a))$.
\begin{proof} (Sketch) 
 Since actor initially have an empty mailbox, if $m(v)\in
  M$, then it must be $Pr\red^* \New{\va'}([b\mapsto N]\, b\{a\s
  m(v)e_b\}\para [a\mapsto M']a\{e'\})\red^*\New{\va}([a\mapsto
  M]a\{e\}\para F)$. By Subject Reduction  
  we know that the actor $b\{a\s m(v);e_b\}$ is well typed, that is
  $\Gamma_b\vdash a:[S_a.\&\marked{}\{?m(T).S,..\}]$, hence
  $?m(T).S\in Inputs(\Gamma_b(a))$. Now, if $?m(T).S\notin
  Inputs(\Gamma(a))$, it means that the input handler in $a$ has been
  consumed by another output that would also require the type
  $a:[S_a.\&\marked{}\{?m(T).S,..\}]$, which is not possible since
  marking is linear. 
\end{proof}

\begin{lemma}\label{lem:NoDeadlock}\mbox{ }\\
Let be $\varnothing\vdash Pr\escape\Delta$ with 
$\balanced(\Delta)$ and $Pr\red^*\New{\va}([a_1\mapsto
M_1]\,a_1\{e_1\}\para \ldots\para [a_k\mapsto M_k]\,
a_k\{e_k\})\not\red$.
Then it is not possible that every actor body $e_i$ is a (stuck) input
expression.
\end{lemma} 
\begin{proof} (Sketch) 
By Subject Reduction there exists $\Delta'$ such that\\
$\varnothing\vdash \New{\va}([a_1\mapsto
M_1]\,a_1\{e_1\}\para \ldots\para [a_k\mapsto M_k]\,
a_k\{e_k\})\escape\Delta'$, hence there exist 
$\Gamma_1,...,\Gamma_k$ such that
$\Gamma_1\odot\ldots\odot\Gamma_k\subseteq\Delta'$ and
$\Gamma_i\vdash [a_i\mapsto M_i]\, a_i\{e_i\}\escape\Delta_i$ for
$i=1,..,k$. By contradiction, assume that every actor body $e_i$ is a
(stuck) input expression. Then we have that for all $i=1,..,k$,  
$\Gamma_i\vdash a_i:[\&\{?m^i_\ell(x^i_\ell).S^i_\ell\}]$. Moreover, 
since the initial system is balanced and since by hypothesis no
matching message is already in the mailbox, for any actor $a_i$ there
must be an actor $a_j$ whose body sends a matching message, i.e. 
$e_j$ must contain the sub-expression $a_i\s m^i_\ell(c)$ for some
message $m_\ell$. However, this is not possible since the typing of
(the output action of) $a_j$ depends on the (continuation of the
input) type of $a_i$ and so on yielding a cyclic dependence between a
set of actors which would require recursive typing.
\end{proof}

\noindent
{\bf Safety Theorem} Let be $\varnothing\vdash Pr\escape\Delta$ with 
$\balanced(\Delta)$. If $Pr\red^* F$ then either $F=\dead$ or $F\red
F'$ for some $F'$.
\begin{proof} (Sketch)
Let prove it by contradiction: assume that there exists a
configuration $F^*\neq\dead$ such that $F^*\not\red$. 
We can assume $F^*\equiv \New{\va}([a_1\mapsto M_1]\,a_1\{e_1\}\para
\ldots\para [a_k\mapsto M_k]\, a_k\{e_k\})$ with
$\{a_1,...,a_k\}\subseteq \va$. Then by Subject Reduction 
there exists $\Delta'$ such that $\varnothing\vdash
F^*\escape\Delta'$, with
$\Delta'\sseq\Delta$ and $\balanced(\Delta')$. Then we also have that
there exist $\Gamma_1,...,\Gamma_k$ such that
$\Gamma_1\odot\ldots\odot\Gamma_k\subseteq\Delta'$ and
$\Gamma_i\vdash [a_i\mapsto M_i]\, a_i\{e_i\}\escape\Delta_i$ for
$i=1,..,k$. 
We have one of the following cases: 
\begin{itemize}
\item for all $i\in\{1,..,k\}$, if $e_i=a_j\s m(\vc);e$ then we
  have that $\Gamma_i\vdash
  a_j:[S_u.\&\marked{}\{\marked{?}m(\vT).S,..\}]$ and by 
  definition of $\odot$ we have that
  $\Gamma_j(a)=[S'_u.\&\{?m(\vT).S,..\}]$ for some $S'_u\sseq S_u$,
  that is the input handler of the message $m$ is still in the body of
  the actor $a_j$, i.e. $e_j\neq\dead$, that is 
  $a_j\in\{a_1,...,a_k\}$ and $F^*\red$ contradicting the assumption.
\item for all $i\in\{1,..,k\}$, if $e_i=\dead$ we have two cases:
  if $M_i=\varnothing$ then the {\sc (Ended)} reduction rule 
  applies, giving a contradiction. On the other hand, let be 
  $M_i\neq\varnothing$, from $e_i=\dead$ and the typing we have
  $\Gamma_i\vdash a_i:[\fineST]$, 
  but by Lemma~\ref{lem:mailbox} and $M_i\neq\varnothing$ we have
  $\Gamma\not\vdash a_i:[\fineST]$, giving the desired contradiction.
\item If an actor body starts with the spawning of a new actor, then
  trivially $F^*\red$, giving the contradiction. 
\item Finally, we have the case where every actor body $e_i$ starts
  with an input expression, which is not possible by
  Lemma~\ref{lem:NoDeadlock}. 
\end{itemize}

\end{proof}

\end{document}